\documentclass[11pt,envcountsame]{llncs}

\usepackage{latexsym,amsfonts}
\usepackage{amsmath}
\usepackage{amssymb}
\usepackage{array}

\usepackage{graphicx,color}
\usepackage{latexsym,amsfonts}
\usepackage{amsmath}
\usepackage{amssymb}
\usepackage{bbm}

\usepackage{algorithmic}
\usepackage{algorithm}
\usepackage{tikz}
\usepackage{url}
\usepackage{fullpage}
\usepackage{enumitem}
\usepackage{environ}
\usepackage{mathdots}

\newcommand{\repeattheorem}[1]{%
  \begingroup
  \renewcommand{\thetheorem}{\ref{#1}}%
  \expandafter\expandafter\expandafter\theorem
  \csname reptheorem@#1\endcsname
  \endtheorem
  \endgroup
}
\NewEnviron{reptheorem}[1]{%
  \global\expandafter\xdef\csname reptheorem@#1\endcsname{%
    \unexpanded\expandafter{\BODY}%
  }%
  \expandafter\theorem\BODY\unskip\label{#1}\endtheorem
}
\usetikzlibrary{lindenmayersystems,calc,decorations,decorations.pathmorphing,decorations.markings,shapes,shapes.geometric,arrows,backgrounds, patterns,decorations.pathreplacing, positioning}

\tikzset{snake it/.style={decorate, decoration=snake}}
\tikzstyle{vertex} = [circle, draw=black, fill=black, inner sep=0pt,  minimum size=5pt]
\tikzstyle{edgelabel} = [circle, fill=white, inner sep=0pt,  minimum size=15pt]

\definecolor{blue}{RGB}{0,82,147} 
\definecolor{red}{RGB}{202,033,063}

\colorlet{green}{green!50!black}
\usepackage[breaklinks=true]{hyperref} 
\hypersetup{
	colorlinks=true,
	citecolor=green,
	linkcolor=blue,  
	urlcolor=black,
}
\newcounter{lpnumber} 

\newcommand{\wt}{\mathsf{wt}}
\newcommand{\cost}{\mathsf{cost}}
\newcommand{\capac}{\mathsf{cap}}
\newcommand{\vote}{\mathsf{vote}}
\newcommand{\nbr}{\mathsf{Nbr}}

\newtheorem{new-claim}{Claim}

\begin{document}
\title{Perfect Matchings and Popularity in the Many-to-Many Setting\thanks{This result appeared in FSTTCS~2023~\cite{KM23}. This version gives a simple proof that corrects a technical gap in the proof for the many-to-many setting in the conference version.}}
\author{Telikepalli Kavitha\inst{1} \and Kazuhisa Makino\inst{2}}
\institute{Tata Institute of Fundamental Research, Mumbai; \email{kavitha@tifr.res.in} \and Research Institute for Mathematical Sciences, Kyoto University, Kyoto; \email{makino@kurims.kyoto-u.ac.jp}}
\maketitle
\pagestyle{plain}

\begin{abstract}
  We consider a matching problem in a bipartite graph $G$ where every vertex has a capacity and a strict preference order
  on its neighbors. Furthermore, there is a cost function on the edge set. We assume $G$ admits a perfect matching, 
  i.e., one that fully matches all vertices. It is only perfect matchings that are feasible for us and we are interested 
  in those perfect matchings that are {\em popular} within the set of perfect matchings. It is known that such matchings 
  (called popular perfect matchings) always exist and can be efficiently computed. 
  What we seek here is not any popular perfect matching, but a {\em min-cost} one. 
  We show a polynomial-time algorithm for finding such a matching; this is via a characterization of popular perfect matchings in 
  $G$ in terms of stable matchings in a colorful auxiliary instance. This is a generalization of such a 
  characterization that was known in the one-to-one setting.
\end{abstract}

\section{Introduction}
\label{sec:intro}
Consider a matching problem in a bipartite graph $G = (A \cup B, E)$ where every vertex $v \in A \cup B$ has a strict ranking of its neighbors. 
For convenience, vertices in $A$ will be called {\em agents} and those in $B$ will be called {\em jobs}. Every agent/job $v$ has an integral capacity 
$\capac(v) \ge 1$ and seeks to be matched to $\capac(v)$ many neighbors. This is the bipartite {\em $b$-matching} framework.

This model is also called the {\em many-to-many} setting in matchings under preferences. This is a well-studied model 
that is a natural generalization of several real-world settings such as matching students to schools and 
colleges~\cite{AS03,BCCKP18} and residents to hospitals~\cite{CRMS,NRMP}. Note that agents (i.e., students/residents) in these applications
have capacity~1 while in the many-to-many setting, agents are allowed to have capacity more than 1. For example, when $A$ is a set of students
and $B$ is a set of projects, the many-to-many setting allows for the flexibility that not only can several students work on one project, but
a student can be part of several projects. One such application is a management system for a conference where program committee members have preferences over submissions and each submission has a ranking of PC members as per their suitability -- in terms of their areas 
of expertise -- to review this paper.\footnote{PC members typically have rankings with ties over papers; we can instead assume a project evaluation committee where each member has a strict ranking of the project proposals that he/she would like to evaluate and each proposal needs a certain number of independent evaluations.}
Every PC member (also, assume every paper) has a capacity.

\begin{definition}
    A matching $M$ in a many-to-many instance $G = (A \cup B, E)$ is a subset of the edge set $E$ such that $|M(v)| \le \capac(v)$ for each vertex $v$, 
    where $M(v) = \{u: (u,v) \in M\}$. 
\end{definition}

Though vertices in both the sets $A$ and $B$ have capacities, note that $M$ can contain at most one copy of any edge since $M$ is a set (and not a multiset).
Let us assume that the input instance admits a perfect
matching, i.e., one that matches each vertex in $A \cup B$ fully up to its capacity. 

Such a model occurs in the conference management system application mentioned above where PC members have to place bids and then they are explicitly asked to bid again for submissions that have received too few bids. Thus it will be ensured that the input instance admits a matching where each PC member $a$ gets assigned to $\capac(a)$ many papers and each paper $b$ is assigned to $\capac(b)$ many PC members.
So we assume our input instance admits a perfect matching and we seek a {\em best} perfect matching as per vertex preferences. In the domain of matchings under preferences, {\em stable} matchings are usually regarded as the best matchings.

\paragraph{\bf Stable matchings.} A matching $M$ is stable if there is no edge that {\em blocks}~$M$, where an edge $(a,b) \notin M$ blocks $M$ 
if (i)~either $a$ has less than $\capac(a)$ partners in $M$ or $a$~prefers $b$ to its worst partner in~$M$ 
and (ii)~either $b$ has less than $\capac(b)$ partners in $M$ or $b$~prefers $a$ to its worst partner in~$M$.
Stable matchings always exist in $G$ and a natural modification of the Gale-Shapley algorithm~\cite{GS62} finds one. However stability is a strong and 
rather restrictive notion, e.g., stable matchings need not have maximum cardinality. 

Consider 
the following instance where $A = \{a,a'\}, B = \{b,b'\}$, and all vertex capacities are~1. The preferences of agents are as follows: 
$a: b \succ b'$ and $a': b$. That is, the top choice of $a$ is $b$ and its second choice is $b'$ while $a'$ has only one neighbor~$b$. Similarly, the
preferences of jobs are as follows: $b: a \succ a'$ and $b': a$.
This instance has only one stable matching $S = \{(a,b)\}$ that leaves $a'$ and $b'$ unmatched. Observe that this instance
has a perfect matching $M = \{(a,b'),(a',b)\}$ that matches all agents and jobs. However $M$ is not stable as the edge $(a,b)$ blocks $M$.

For the problem of selecting a {\em best} perfect matching, one with the least number of blocking edges would be a natural relaxation
of stability. But 
finding such a matching is NP-hard~\cite{BMM10} even in the one-to-one setting.
A well-studied relaxation of stability that offers a meaningful and tractable solution to the problem of finding a ``best perfect matching''
is the notion of {\em popularity}. 

\paragraph{\bf Popularity.}
Popularity is based on voting by vertices on matchings.
In the one-to-one setting, the preferences of a vertex over its neighbors extend naturally to preferences 
over matchings---so every vertex orders matchings in the order of its partners in these matchings. 
Popular matchings are weak {\em Condorcet winners}~\cite{Con85,condorcet} in this voting instance where vertices are voters and matchings are candidates.
In other words, a popular matching $M$ does not lose a head-to-head election against any matching~$N$
where each vertex $v$ either casts a vote for the matching in $\{M,N\}$ that it prefers or $v$ abstains from voting if its assignment is the same in 
$M$ and $N$. 

For any vertex $v$ and its neighbors $u$ and $u'$, let $\vote_v(u,u')$ be $v$'s vote for $u$ versus $u'$. 
More precisely, $\vote_v(u,u') = 1$ if $v$ prefers $u$ to $u'$, it is $-1$ if $v$ prefers $u'$ to $u$, and it is 0 otherwise (so $u = u'$).
Thus in the one-to-one setting, every vertex $v$ casts its vote which is $\vote_v(u,u') \in \{0, \pm 1\}$, where $M(v) = \{u\}$ and $N(v) = \{u'\}$, 
for $M$ versus $N$ in their head-to-head election.

Recall that we are in the many-to-many setting, i.e., vertices have capacities. So we need to specify how a vertex votes over different 
subsets of its neighbors. Thus we need to compare two subsets $M(v)$ and $N(v)$ of $\nbr(v)$ for any vertex $v$, where $\nbr(v)$ is the set 
of neighbors of $v$. We will follow the method from \cite{BK19} for this comparison where
$v$ is allowed to cast up to $\capac(v)$ many votes. Any two subsets $S$ and $T$ of $\nbr(v)$ are compared by vertex $v$ as follows:
\begin{itemize}
    \item a bijection $\psi$ is chosen from $S'= S \setminus T$ to $T' = T \setminus S$;\footnote{If the two sets are not of equal size, then dummy vertices that are less preferred to all non-dummy vertices are added to the smaller set.}
    \item every neighbor $u \in S'$ is compared with $\psi(u) \in T'$;
    \item the number of wins minus the number of losses is $v$'s vote for $S$ versus $T$.
\end{itemize}

The bijection $\psi$ from $S'$ to $T'$ that is chosen will be the one that \emph{minimizes} $v$'s vote for $S$ versus~$T$. 
More formally, the vote of $v$ for $S$ versus $T$, denoted by $\vote_v(S,T)$, is defined
as follows where $|S'| = |T'| = k$ and $\Pi[k]$ is the set of permutations on $\{1,\ldots,k\}$:
\begin{equation}
  \label{eq:delta}
         \vote_v(S,T) = \min_{\sigma \in \Pi[k]}\sum_{i = 1}^k \vote_v(S'[i],T'[\sigma(i)]).
\end{equation}
Here $S'[i]$ is the $i$-th ranked resident in $S'$ and $T'[\sigma(i)]$ is the $\sigma(i)$-th ranked resident in $T'$.
Consider the following example from \cite{BK19} where $\nbr(v) = \{u_1,u_2,u_3,u_4,u_5,u_6\}$ and $v$'s preference order is
$u_1 \succ u_2 \succ u_3 \succ u_4 \succ u_5 \succ u_6$. 
Let $S = \{u_1,u_3,u_5\}$ and $T = \{u_2,u_4,u_6\}$.  So $S' = S$ and $T' = T$.
The bijection $\psi: S' \rightarrow T'$ is $\psi(u_1) = u_6$, $\psi(u_3) = u_2$, and $\psi(u_5) = u_4$ as this is the
most adversarial way of comparing $S'$ with $T'$. This results in $\vote_v(S,T) = \vote_v(u_1,u_6) + \vote_v(u_3,u_2) + \vote_v(u_5,u_4) = 1 - 1 - 1 = -1$.
Similarly, $\vote_v(T,S) = \vote_v(u_2,u_1) + \vote_v(u_4,u_3) + \vote_v(u_6,u_5) = -3$.

For any pair of matchings $M,N$ and vertex $v$, let $\vote_v(M,N) = \vote_v(S,T)$ where $S = M(v)$ and $T = N(v)$.
So $\vote_v(M,N)$ counts the number of votes by $v$ for $M(v)$ versus $N(v)$ 
when the two sets $M(v)\setminus N(v)$ and $N(v)\setminus M(v)$ are compared 
in the order that is most adversarial or {\em negative} for~$M$. 
The two matchings $M$ and $N$ are compared using $\Delta(M,N) = \sum_{v\in A \cup B} \vote_v(M(v),N(v))$.
We will say matching~$M$ is more popular than matching $N$ if $\Delta(M,N) > 0$.

\begin{definition}
\label{def:pop-mat}
A matching $M$ is popular if $\Delta(M,N) \ge 0$ for all matchings $N$ in $G$.
\end{definition}

Thus $M$ is popular if there is no matching more popular than $M$.
Recall our assumption that the set of feasible solutions is the set of 
perfect matchings. 
Hence the matchings of interest to us are {\em popular perfect matchings}, defined below.

\begin{definition}
\label{def:pop-max-mat}
    A perfect matching $M$ is a popular perfect matching if $\Delta(M,N) \ge 0$ for all perfect matchings $N$ in $G$.
\end{definition}

So a popular perfect matching need not be popular---it may lose to a smaller size matching; nevertheless, it never loses to a
perfect matching. Since weak Condorcet winners do not always exist, it is not a priori clear if a popular perfect matching always exists.
It was shown in \cite{Kav14} that there always exists a popular perfect matching in the one-to-one setting;\footnote{It was shown in \cite{Kav14} that
a popular {\em maximum} matching (i.e., a maximum matching that is popular among maximum matchings) always exists in the one-to-one setting.}  
moreover, such a matching can be computed in polynomial time. 

The following problem in the many-to-many setting was considered in~\cite{NNRS22}: vertices have capacity lower bounds and feasible matchings are those
that satisfy these lower bounds. 
It was shown that popular feasible matchings always exist and can be computed in polynomial time. Popular perfect matchings are a special case of this
problem where every vertex has to be fully matched, i.e., every vertex~$v$ has a lower bound of $\capac(v)$. Thus popular perfect matchings always exist in the
many-to-many setting and can be computed in polynomial time.

\paragraph{\bf Our problem.}
We assume there is a function $\cost: E \rightarrow \mathbb{R}$, which is part of the
input. Hence the cost of any matching $M$ is $\sum_{e\in M}\cost(e)$. 
There might be exponentially many popular perfect matchings in $G$, hence we would like to find an {\em min-cost} one, i.e.,
a popular perfect matching whose cost is least among all popular perfect matchings. 
Solving the min-cost popular perfect matching problem efficiently implies efficient 
algorithms for several desirable popular perfect matching problems such as finding one with the highest utility
when every edge has an associated utility or one with forced/forbidden edges
or an egalitarian one. In the conference management system application discussed earlier, we would like to find a popular 
perfect matching that matches as many PC members as possible along top ranked edges, subject to that, as many PC members 
as possible along second ranked edges, and so on~\cite{IKMMP06}. Such a popular perfect matching is
a min-cost popular perfect matching for an appropriate cost function. 

A polynomial-time algorithm for the min-cost popular perfect matching problem in the one-to-one setting was shown in \cite{Kav24}.
The conference version of our paper~\cite{KM23} gave a polynomial-time algorithm for this problem in the {\em many-to-one} setting (i.e., $\capac(a) = 1$ for
all $a \in A$) by reducing it to the min-cost popular perfect matching 
problem in the one-to-one setting and using the polynomial-time algorithm for computing a min-cost popular perfect matching in the one-to-one setting~\cite{Kav24}.

\paragraph{\bf The hospitals/residents setting.}
The many-to-one setting is usually referred to as the hospitals/residents setting.
The reduction in \cite{KM23} from the hospitals/residents setting to the one-to-one setting is via the {\em cloned} instance $G' = (A \cup B', E')$ corresponding to the 
given hospitals/residents instance $G = (A \cup B, E)$. Every $b \in B$ is replaced by $\capac(b)$ many clones $b_1,b_2,\ldots$ in $B'$ and the preference order of
every $b_i$ is the same as the preference order of $b$. Furthermore, every $a \in A$ that is a neighbor of $b$ in $G$ replaces the occurrence of $b$ in its
preference order with $b_1 \succ \cdots \succ b_{\capac(b)}$. 

There is a natural map $f$ from the set of popular perfect matchings in the one-to-one instance $G'$ to the set of perfect
matchings in the original instance $G$, where for any popular perfect matching~$M'$ in $G'$, the many-to-one matching $f(M')$ is 
obtained by replacing each edge $(a,b_i)$ in $M'$ with the original edge $(a,b)$. It can be shown that $f(M')$ is a popular perfect matching in $G$.  
So we have $f: \{$popular perfect matchings in $G'\} \rightarrow$ $\{$popular perfect matchings in $G\}$. 
It was shown in \cite{KM23} that the mapping $f$ is surjective.
Hence solving the min-cost popular perfect matching problem in $G'$ solves the min-cost popular perfect matching problem in~$G$.
Though it was claimed in \cite{KM23} that the same approach can be generalized to the many-to-many setting, unfortunately
it does not extend to the many-to-many setting. 

\paragraph{\bf The many-to-many setting.}
Consider $G = (A \cup B, E)$ where $A = \{a,a'\}$ and $B = \{b,b'\}$ and every vertex has capacity~2. 
Vertex preferences are as follows: $a\colon b \succ b'$ and $a'\colon b' \succ b$ along with
$b\colon a \succ a'$ and $b'\colon a' \succ a$.

There is only one perfect matching $M = \{(a,b),(a,b'),(a',b),(a',b')\}$ in $G$,
hence $M$ is a popular perfect matching. But $M$ cannot be realized as a popular perfect matching in 
the corresponding one-to-one instance $G'$.  The vertex set of $G' = (A' \cup B', E')$ is given by 
$A' = \{a_1,a_2,a'_1,a'_2\}$ and $B' = \{b_1,b_2,b'_1,b'_2\}$. Vertex preferences in $G'$ are as follows.

\begin{minipage}[c]{0.45\textwidth}
			\centering
			\begin{align*}
				&a_1\colon b_1 \succ b_2 \succ b'_1 \succ b'_2 \qquad\qquad && b_1\colon a_1 \succ a_2 \succ a'_1 \succ a'_2\\
                &a_2\colon b_1 \succ b_2 \succ b'_1 \succ b'_2 \qquad\qquad && b_2\colon a_1 \succ a_2 \succ a'_1 \succ a'_2\\
                &a'_1\colon b'_1 \succ b'_2 \succ b_1 \succ b_2 \qquad\qquad && b'_1\colon a'_1 \succ a'_2 \succ a_1 \succ a_2\\
                &a'_2\colon b'_1 \succ b'_2 \succ b_1 \succ b_2 \qquad\qquad && b'_2\colon a'_1 \succ a'_2 \succ a_1 \succ a_2\\
			\end{align*}
		\end{minipage}

Observe that $S' = \{(a_1,b_1),(a_2,b_2),(a'_1,b'_1),(a'_2,b'_2)\}$ is a stable matching in $G'$. Hence it is popular~\cite{Gar75}, so $S'$ is a popular perfect matching in $G'$.
However the matching $\{(a,b),(a',b')\}$ obtained by replacing edges in $G'$ with original edges is
{\em not} perfect in $G$. 
Note that we are not allowed to have {\em two} copies of any edge here, hence both $(a,b)$ and $(a',b')$ are present with multiplicity~1. Thus popular perfect matchings in $G$ and $G'$ are quite different from each other.

\paragraph{\bf Our approach.}
We need a different approach to solve the min-cost popular perfect matching problem in the many-to-many setting. Our approach is to reduce the min-cost
popular perfect matching problem in $G = (A \cup B, E)$ to the min-cost
stable matching problem in an auxiliary many-to-many instance $G^* = (A \cup B, E^*)$. This instance $G^*$ is a multigraph, i.e., parallel edges
are present in $G^*$.
A min-cost stable matching in a many-to-many instance  can be computed in polynomial time~\cite{Fle03,huang2010}.
Furthermore, Fleiner's algorithm~\cite{Fle03} works even in the case when the input instance is a multigraph. 
Hence a min-cost popular perfect matching in the original many-to-many instance $G$ can be computed in polynomial time. 
Thus we show the following result.

\begin{reptheorem}{restate_thm:many-many}
  \label{thm:first}
Let  $G = (A \cup B, E)$ be a many-to-many matching instance where vertices have strict preferences (possibly incomplete) and there is
a function $\cost: E \rightarrow \mathbb{R}$.
If $G$ admits a perfect matching then a min-cost popular perfect matching in~$G$ can be computed in polynomial time. 
\end{reptheorem}

\subsection{Background and Related Results}
\label{sec:background}
The notion of popularity was proposed by G\"ardenfors~\cite{Gar75} in 1975 in the stable {\em marriage} problem
(i.e., the one-to-one setting) where he observed that stable matchings are popular. It was shown
in \cite{BIM10,CHK15} that when preferences include ties (even one-sided ties), it is NP-hard to decide if a popular matching exists or not.
It was shown in \cite{HK11} that every stable matching in a marriage instance is a min-size popular matching.
Polynomial-time algorithms to find a max-size popular matching were shown in \cite{HK11,Kav14}. 
We refer to \cite{Cseh} for a survey on results in popular matchings in the one-to-one setting. 

The notion of popularity was extended from one-to-one setting to many-to-many setting in \cite{BK19} and \cite{NR17}, independently. 
A polynomial-time algorithm to compute a max-size popular matching in the many-to-many setting was given in \cite{BK19}.
It was also shown in \cite{BK19} that every stable matching in the many-to-many setting is popular; so though a rather strong definition 
of popularity was adopted here, popular matchings always exist. The definition of popularity considered in \cite{NR17} is weaker than the one 
in \cite{BK19}; in order to compare a pair of matchings $M$ and $N$, every vertex $v$ uses the bijection that compares the top neighbor in 
$M(v)\setminus N(v)$ with the top neighbor in $N(v)\setminus M(v)$, and so on, i.e., the permutation $\sigma$ in Eq.~\eqref{eq:delta} is the 
identity permutation. The max-size popular matching problem with matroid constraints (this model generalizes popular many-to-many matchings) was considered in
\cite{CKY24,Kam20} and shown to be tractable. Strongly popular matchings 
(such a matching defeats every other matching) in many-to-many instances were studied in \cite{KM20}.

The stable matching problem has been extensively studied in the hospitals/residents and the many-to-many settings~\cite{askalidis2013,Bla88,Fle03,huang2010,HIM16,IMS00,IMS03,Roth84b,Sot99} and 
a min-cost stable matching in the one-to-one setting (and thus in the hospitals/residents setting) can be computed in polynomial time~\cite{Rot92,TS98}.
Note that min-cost stable matching algorithms in the one-to-one setting do not generalize to the many-to-many setting. This is because
(unlike the hospitals/residents setting) the many-to-many setting cannot be reduced to the one-to-one setting via cloning~\cite[Footnote~6]{huang2010}. 

The algorithms in \cite{Fle03,huang2010} solve the min-cost stable matching problem in the many-to-many setting.
Fleiner's algorithm~\cite{Fle03} solves the min-cost {\em matroid kernel} problem in the intersection of two strictly ordered matroids.
This generalizes the min-cost stable matching problem in the many-to-many setting. Huang's algorithm~\cite{huang2010} solves the
min-cost {\em classified} stable matching problem when each vertex on one side of the graph has {\em classifications} with upper bounds on the
number of partners it can have in each class. Note that this problem generalizes the min-cost stable matching problem in the many-to-many setting. 
The min-cost classified stable matching problem when vertices on both sides of the graph have classifications was solved in \cite{FK16}.

There is a polynomial-time algorithm to find a min-cost popular maximum matching in the one-to-one setting~\cite{Kav24}. In contrast to this, 
finding a min-cost popular matching in the one-to-one setting is NP-hard~\cite{FKPZ18}. Nevertheless, when preferences are complete, there is a 
polynomial-time algorithm to find a min-cost popular matching in the one-to-one setting~\cite{CK18}. This result was recently extended to the 
hospitals/residents setting in \cite{KM24}; note that this is a non-trivial extension as the set of popular matchings in a hospitals/residents 
instance can be richer than those in the corresponding one-to-one instance. 

The main result in the conference version of our paper~\cite{KM23} (now included in \cite{KM24}) 
showed that in contrast to popular matchings, the set of popular {\em perfect} matchings in a hospitals/residents
instance is {\em not} richer than the set of popular perfect matchings in the corresponding one-to-one instance.
That is, each popular perfect matching in a hospitals/residents instance can be realized as a popular perfect matching in the corresponding one-to-one instance. As seen in our example earlier, this is not the case for
popular perfect matchings in the many-to-many setting. 

Interestingly, the tractability of a matching problem in the one-to-one setting does not always imply its tractability in the hospitals/residents setting. One such problem is that of finding a matching that maximizes {\em Nash social welfare}
(i.e., the geometric mean of edge utilities).
Such a matching can be easily found by the maximum weight matching algorithm in the one-to-one setting,\footnote{For any edge $(a,b)$, the weight of this edge is $\log$ of the product of utilities that $a$ and $b$ have for each other.} however it is NP-hard to find such a matching in the hospitals/residents setting~\cite{JV24}.

\subsection{Techniques}
\label{sec:techniques}
We will use the characterization of popular perfect matchings in the one-to-one setting. It was shown in
\cite{Kav14} that in order to find a popular maximum matching in a given one-to-one instance, it suffices to run the Gale-Shapley algorithm in an 
appropriate auxiliary instance (a multigraph). It was shown in \cite{Kav24} that this mapping (essentially, a projection) from the set of stable matchings 
in this auxiliary instance to the set of popular maximum matchings in the original instance is surjective. 
Thus for every popular maximum matching in the original instance, there exists a corresponding stable matching in the auxiliary instance.
Hence in the one-to-one setting, 
the min-cost stable matching algorithm in the auxiliary instance solves the min-cost popular maximum matching problem in the original instance.

Our problem is to find a min-cost popular perfect matching in a many-to-many instance $G$. We will show a many-to-many instance $G^*$ (as before, a multigraph)
such that there is a surjective mapping from the set of stable matchings in $G^*$ to the set of popular perfect matchings in $G$. 
Let $M$ be any perfect matching in $G$. We will show in Section~\ref{sec:one-one} that $M$ is a popular perfect matching in $G$ if and only if there is a realization $M'$ of $M$ such that $M'$ is 
a popular perfect matching in an appropriate one-to-one instance---this instance is a subgraph of the one-to-one cloned instance $G'$ described earlier. 
This subgraph depends on the matching $M'$, so let us call it $G'_{M'}$. Thus this reduction from the many-to-many setting to the one-to-one setting
does not seem particularly helpful since it depends on the matching $M$ that we seek. 

We will overcome the above bottleneck in Section~\ref{sec:stable} by using the result for one-to-one instances in \cite{Kav24} to construct a one-to-one multigraph~$G^0_M$ such that $M'$ can be realized as a stable matching in $G^0_M$. 
We will then use this to show a many-to-many multigraph $G^*$ such that 
$M$ can be realized as a stable matching $M^*$ in the many-to-many instance $G^*$.
Conversely, every stable matching in the instance $G^*$ projects to a popular perfect matching in $G$. Thus finding a min-cost stable matching in~$G^*$ (by Fleiner's algorithm~\cite{Fle03}) solves the min-cost popular perfect matching problem in $G$.

\section{A characterization of popular perfect matchings}
\label{sec:one-one}

This section contains our reduction from the popular perfect matching problem in a many-to-many instance $G = (A \cup B, E)$ 
to the popular perfect matching problem in a one-to-one instance---recall that this one-to-one instance will depend on the matching that we seek.

As mentioned in Section~\ref{sec:intro}, corresponding to the many-to-many instance $G$, there is a one-to-one instance 
$G' = (A'\cup B', E')$ obtained via cloning.
So $A' = \{a_i: a \in A \ \text{and}\ 1 \le i \le \capac(a)\}$, $B' = \{b_j: b \in B \ \text{and}\ 1 \le j \le \capac(b)\}$,
and $E' = \{(a_i,b_j): (a,b) \in E\ \text{and}\ 1 \le i \le \capac(a), 1 \le j \le \capac(b)\}$.
Each vertex in $G'$ has capacity~1 and a strict preference order over its neighbors. 
\begin{itemize}
    \item For any vertex $v$ and $i \in \{1,\ldots,\capac(v)\}$: the preference order of $v_i$ in $G'$ is the same as $v$'s preference order in $G$ where every neighbor $u$ of $v$ gets replaced by all its clones in the order $u_1 \succ \cdots \succ u_{\capac(u)}$. 
\end{itemize}

Let $M$ be a perfect matching in $G$ and let $M'$ be any {\em one-to-one realization} of $M$ in $G'$. A one-to-one realization $M'$
of a many-to-many matching $M$ is obtained as follows.
\begin{itemize}
    \item For each edge $(a,b) \in M$, choose an index $i \in \{1,\ldots,\capac(a)\}$ and an index $j \in \{1,\ldots,\capac(b)\}$ such that 
    the indices $i$ and $j$ have not been chosen so far and include the edge $(a_i,b_j)$ in $M'$.
\end{itemize}
Thus $M'$ is a perfect matching in the one-to-one instance $G'$.
The following subgraph $G'_{M'} = (A' \cup B', E'_{M'})$ of $G'$ (originally from \cite{BK19})
will be useful to us. Its edge set $E'_{M'}$ is as follows.
\[ E'_{M'} = M' \cup \{(a_i,b_j): (a,b)\in E\setminus M\ \text{where}\ 1 \le i \le \capac(a)\ \text{and}\ 1 \le j \le \capac(b)\}.\]
For any edge $(a,b) \notin M$, all its $\capac(a)\cdot\capac(b)$ many copies $(a_1,b_1),\ldots,(a_{\capac(a)},b_{\capac(b)})$ 
are in $E'_{M'}$ while for each edge $(a,b) \in M$, exactly {\em one} edge, which is the edge in $M'$ -- say, $(a_i,b_j)$ -- is in $E'_{M'}$. 

We will use the following edge weight function $\wt_{M'}$ in the one-to-one instance $G'_{M'}$. 
Recall the $\vote_v$-function defined in Section~\ref{sec:intro} for a vertex $v$.
For any edge $e = (a,b)$ in $G'_{M'}$, let $\wt_{M'}(a,b) = \vote_a(b,M'(a)) + \vote_b(a,M'(b))$. Thus for any $e \in E'_{M'}$:
\begin{equation*} 
\mathrm{let}\ \wt_{M'}(e) = \begin{cases} 2   & \text{if\ $e$\ blocks\ $M'$;}\\
	                     -2 &  \text{if\ the\ endpoints\ of\ $e$\ prefer\ their\ partners\ in\ $M'$\ to\ each\ other;}\\			
                              0 & \text{otherwise.}
\end{cases}
\end{equation*}

Note that the above weight function has been defined so that for any perfect matching $N'$ in~$G'_{M'}$,
we have $\wt_{M'}(N') = \Delta(N',M') = -\Delta(M',N')$.

The following theorem gives a characterization of popular perfect matchings in a many-to-many instance $G$.
Popular matchings in a many-to-many instance were characterized in \cite[Theorem~2.1]{KM20}. Theorem~\ref{prop2}
is a simpler characterization since it deals with popular {\em perfect} matchings.
This was proved in the hospitals/residents setting in \cite{KM23}; the proof of the ``only if'' part in the many-to-many setting
(given below) is more involved than the proof in the hospitals/residents setting.

\begin{theorem}
    \label{prop2}
    A perfect matching $M$ is a popular perfect matching in a many-to-many instance $G = (A \cup B, E)$ 
    if and only if there is a realization $M'$ of $M$ such that there is no alternating cycle~$C$ with respect to $M'$ in 
    the one-to-one instance $G'_{M'} = (A' \cup B', E'_{M'})$ 
    such that $\wt_{M'}(C) > 0$. 
\end{theorem}
\begin{proof}
    Suppose $M'$ is a realization of $M$ that satisfies the property that
$\wt_{M'}(C) \le 0$ for all alternating cycles $C$ in $G'_{M'}$. 
Let $N$ be any perfect matching in $G$. 
Our first goal is to show a realization $N'$ of $N$ in the one-to-one instance $G'_{M'}$ such that $\Delta(M',N') = \sum_{v\in A\cup B}\vote_v(M,N)$.
Then we will have $\wt_{M'}(N') = -\sum_{v\in A\cup B}\vote_v(M,N)$ because $\wt_{M'}(N') = \Delta(N',M') = -\Delta(M',N')$ by the definition of $\wt_{M'}$. 
Such a realization $N'$ of $N$ can be obtained in $G'_{M'}$ as follows. 

 \begin{itemize}
     \item For every edge $(a,b) \in N \cap M$: the edge $(a_i,b_j)$ is in $N'$ where $(a_i,b_j) \in M'$.
     \item For every $(a,b) \in N\setminus M$: 
     \smallskip
       \begin{itemize}
           \item In the evaluation of $\vote_a(M,N)$ while comparing $M(a) \setminus N(a)$ with $N(a) \setminus M(a)$, 
            let $b'$ be the hospital that $a$ compares $b$ with. So the matching $M'$ contains the edge $(a_i,b'_{j'})$ for some 
            $i \in \{1,\ldots,\capac(a)\}$ and $j' \in \{1,\ldots,\capac(b')\}$. 
           \item In the evaluation of $\vote_b(M,N)$ while comparing $M(b) \setminus N(b)$ with $N(b) \setminus M(b)$, 
           let $a'$ be the agent that job $b$ compares $a$ with. So the matching $M'$ contains the edge 
           $(a'_{i'},b_j)$ for some $i' \in \{1,\ldots,\capac(a')\}$ and $j \in \{1,\ldots,\capac(b)\}$. 
            \smallskip
           \item The edge $(a_i,b_j)$ will be included in $N'$.
          \end{itemize}         
\end{itemize}

Thus $\sum_{v\in A\cup B}\vote_v(M,N) = -\wt_{M'}(N') = -\sum_{C} \wt_{M'}(C)$ where the summation is over all alternating cycles $C$ 
in $M' \oplus N'$. Since $M'$ and $N'$ are perfect matchings in $G'_{M'}$, note that their symmetric difference $M' \oplus N'$ contains
only alternating cycles. Since $\wt_{M'}(C) \le 0$ for all alternating cycles $C$ in $M' \oplus N'$, we have 
$\Delta(M,N) = \sum_{v\in A\cup B}\vote_v(M,N) \ge 0$.
Thus $\Delta(M,N) \ge 0$ for any perfect matching~$N$, hence $M$ is a popular perfect matching in $G$.

\smallskip

Conversely, let $M$ be a popular perfect matching in $G$ and let $M'$ be any realization of $M$ in~$G'$. 
Let $C$ be an alternating cycle with respect to $M'$ in $G'_{M'}$.
Let us call $C$ {\em valid} if the collection of edges in $G$ that correspond to edges in 
$C$ is not a multiset. Not every alternating cycle $C$ 
with respect to $M'$ in $G'_{M'}$ is necessarily valid since  
$C$ may consist of multiple edges $(a_i,b_{i'}), \ldots, (a_j,b_{j'})$ 
that correspond to the same edge $(a,b)$ in $G$. The following claim will be very useful.

\begin{claim}
    If there exists an alternating cycle $C$ with respect to $M'$ in $G'_{M'}$ such that $\wt_{M'}(C) > 0$ then there exists a valid alternating cycle
    $C'$ with respect to $M'$ in $G'_{M'}$ such that $\wt_{M'}(C') > 0$.
\end{claim}

We will assume the above claim and complete this proof. Then we will prove the above claim.
Suppose there is an alternating cycle $C$ with respect to $M'$ in $G'_{M'}$ such that $\wt_{M'}(C) > 0$. 
Then the above claim tells us that there is a valid alternating cycle $C'$ such that $\wt_{M'}(C') > 0$. 
Moreover, $M' \cup C'$ does not become a multiset with respect to $E$ since $E'_{M'} \supseteq M' \cup C'$ has only one copy of edges in $M'$. 
 
 Let $N' = M' \oplus C'$. Observe that $\wt_{M'}(N') = \wt_{M'}(C') > 0$. 
    Since $M' \cup C'$ is not a multiset, by identifying all clones of the same vertex, 
    the matching $N'$ in $G'_{M'}$ becomes a perfect matching $N$ in $G$. 
For any vertex $v$ with no clone in $C'$, the two sets $M(v)$ and $N(v)$ are the same.
For each vertex $v$ with one or more clones in $C'$, the cycle $C'$ defines a bijection between $M(v)\setminus N(v)$ and $N(v)\setminus M(v)$.   
Since $\wt_{M'}(C') > 0$, this means for each vertex $v$ in $C'$, there is a way of comparing elements in $M(v)\setminus N(v)$ with those in $N(v)\setminus M(v)$ 
such that summed over all vertices in~$C'$, the votes in favor of $N$ outnumber the votes in favor of $M$. 

Thus adding up the votes of all vertices
in $C'$ for $M$ versus $N$ (as per the bijection defined by~$C'$), the total number of votes for $M$ is less than the total 
number of votes for $N$. Hence it follows from the definition of the function $\vote$ (see Eq.\eqref{eq:delta}) that $\sum_{v\in C'}\vote_v(M,N) < 0$.
So $\Delta(M,N) < 0$ and this contradicts the fact that $M$ is a popular perfect matching in $G$. \qed
\end{proof}

\begin{proof}[of the above claim]
Consider the shortest alternating cycle $C'$ with respect to $M'$ in $G'_{M'}$ such that $\wt_{M'}(C') > 0$.
For each edge in $M$, there is precisely one copy of this edge in $G'_{M'}$. So
if $C'$ is not valid, then it has to be the case that $C' \setminus M'$ contains more than one copy of the same edge. Let $(a,b)$ 
be such an edge. Since $(a,b) \notin M$, it follows from the construction of $G'_{M'}$ that $(a_i,b_j) \in E'_{M'}$ 
for every $i \in \{1,\ldots,\capac(a)\}$ and $j \in \{1,\ldots,\capac(b)\}$. 

Let $(a_1,b_1)$ and $(a_2,b_2)$ be two copies of $(a,b)$ in $C'$. 
Let us construct the following two shorter alternating cycles $C_1$ and $C_2$ 
by replacing the edges $(a_1,b_1)$ and $(a_2,b_2)$ in $C'$ with the edges
$(a_1,b_2)$ and $(a_2,b_1)$. This can be visualized as drawing two chords $(a_1,b_2),(a_2,b_1)$ 
inside $C'$ to create two smaller alternating cycles $C_1$ and $C_2$ by using the edges in 
$\{(a_1,b_2),(a_2,b_1)\} \cup C'\setminus\{(a_1,b_1),(a_2,b_2)\}$ (see Fig.~\ref{fig:example}). 

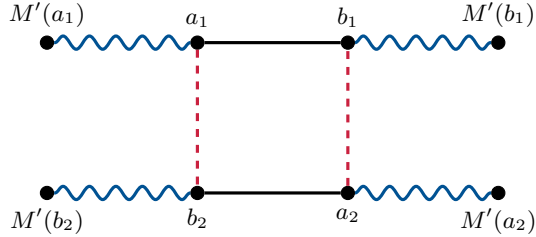
\begin{figure}[h]
\tikzstyle{vertex} = [circle, draw=black, fill=black, inner sep=0pt,  minimum size=5pt]
\tikzstyle{edgelabel} = [circle, fill=white, inner sep=0pt,  minimum size=15pt]
\centering
	\pgfmathsetmacro{\d}{2}
\begin{minipage}{0.4\textwidth}
\begin{tikzpicture}[scale=1, transform shape]
	\node[vertex, label=above:$a_1$] (a1) at (0,0) {};
	\node[vertex, label=below:$b_2$] (b1) at ($(a1) + (0, -\d)$) {};
	\node[vertex, label=above:$b_1$] (b2) at (\d,0) {};
	\node[vertex, label=below:$a_2$] (a2) at ($(b2) + (0, -\d)$) {};
	\node[vertex, label=above:$M'(a_1)$] (b3) at (-\d,0) {};
	\node[vertex, label=above:$M'(b_1)$] (a3) at ($(a1) + (2*\d, 0)$) {};
        \node[vertex, label=below:$M'(b_2)$] (a4) at ($(b3) + (0, -\d)$)  {};
         \node[vertex, label=below:$M'(a_2)$] (b4) at ($(a3) + (0, -\d)$)  {};

         \draw [very thick,red,dashed] (a1) -- (b1);
         \draw [very thick,red,dashed] (a2) -- (b2);
         \draw [very thick] (a1) -- (b2);
         \draw [very thick] (a2) -- (b1);
        \draw [very thick, blue, snake it] (a1) -- (b3);
         \draw [very thick, blue, snake it] (b1) -- (a4);
          \draw [very thick, blue, snake it] (a2) -- (b4);
           \draw [very thick, blue, snake it] (b2) -- (a3);
\end{tikzpicture}
\end{minipage}

\caption{The matching $M'$ is described by wavy blue edges. We are deleting the edges $(a_1,b_1)$ and $(a_2,b_2)$ from $C'$ and adding the dashed red edges $(a_1,b_2)$ and $(a_2,b_1)$ to form two shorter alternating cycles $C_1$ and $C_2$.}
\label{fig:example}
\end{figure}

Consider the sum $\wt_{M'}(a_1,b_1) + \wt_{M'}(a_2,b_2)$. This equals: 
\begin{eqnarray*}
            & \ & \ \vote_{a_1}(b_1, M'(a_1)) + \vote_{b_1}(a_1,M'(b_1)) + \vote_{a_2}(b_2,M'(a_2)) + \vote_{b_2}(a_2,M'(b_2))\\
            & = & \ \vote_{a_1}(b_2,M'(a_1)) + \vote_{b_1}(a_2,M'(b_1)) + \vote_{a_2}(b_1,M'(a_2)) + \vote_{b_2}(a_1,M'(b_2))\\
            & = & \ \vote_{a_1}(b_2,M'(a_1)) + \vote_{b_2}(a_1,M'(b_2)) + \vote_{a_2}(b_1,M'(a_2)) + \vote_{b_1}(a_2,M'(b_1))\\
            & = & \ \wt_{M'}(a_1,b_2) + \wt_{M'}(a_2,b_1).   
\end{eqnarray*}

Note that these equalities crucially use the structure of vertex preferences in the cloned graph $G'$. 
It follows from the above equalities that $\wt_{M'}(C') = \wt_{M'}(C_1) + \wt_{M'}(C_2)$. 
Since $\wt_{M'}(C') > 0$, it has to be the case that $\wt_{M'}(C_1) > 0$ or 
$\wt_{M'}(C_2) > 0$ (or possibly both). Since both $C_1$ and $C_2$ are shorter than $C'$,
this contradicts the definition of $C'$ as a shortest alternating cycle with positive weight under $\wt_{M'}$. 
Hence $C'$ has to be valid. \qed
\end{proof}

Theorem~2.1 in \cite{Kav24} showed that a perfect matching $M'$ in a one-to-one instance is a
popular perfect matching if and only if there is no alternating cycle $C$ with respect to $M'$ such that $\wt_{M'}(C)>0$.
Corollary~\ref{lem:char} follows from Theorem~\ref{prop2} and \cite[Theorem~2.1]{Kav24}. 

\begin{corollary}
  \label{lem:char}
  Let $M$ be any perfect matching in $G = (A \cup B, E)$. Then $M$ is a popular perfect matching in $G$ if and only if there is a realization 
  $M'$ of $M$ such that $M'$ is a popular perfect matching in the one-to-one instance $G'_{M'} = (A' \cup B', E'_{M'})$.
\end{corollary}

\section{A stable matching problem}
\label{sec:stable}

Let $M$ be a popular perfect matching in $G$. Section~\ref{sec:one-one} showed that 
for any one-to-one realization $M'$ of $M$, the matching
$M'$ is a popular perfect matching in $G'_{M'}$.
This fact will be used in this section to show the existence of a {\em colorful} $M$ (call it $M^*$) such that $M^*$ 
is a stable matching in a colorful many-to-many instance $G^*$. We will first construct a colorful one-to-one instance in Section~\ref{sec:color-one} and the colorful many-to-many instance $G^*$ will be constructed in Section~\ref{sec:color-many}.

\subsection{A colorful one-to-one instance} 
\label{sec:color-one}
Let $M$ be a perfect matching in $G$ and let $M'$ be any realization of $M$.
Corresponding to the instance $G'_{M'} = (A' \cup B', E'_{M'})$, we will construct the following one-to-one instance 
$G^0_{M'} = (A' \cup B', E^0_{M'})$.
The edge set $E^0_{M'}$ of this graph is  $\{e_c: e\in E'_{M'}$ and $1 \le c \le n_0\}$
where $n_0 = |A'|$. That is, every edge $e = (a_i,b_j)$ in $G'_{M'}$ has $n_0$ parallel copies $e_1,\ldots,e_{n_0}$ in $G^0_{M'}$. 
Thus $G^0_{M'}$ is a multigraph. As done in other papers, it will be convenient to visualize the $n_0$ copies of any edge $e$, i.e., $e_1,\ldots,e_{n_0}$, 
as $e$ colored by $n_0$ different colors: $\text{color}~1,\ldots,$  $\text{color}~n_0$ (see Fig.~\ref{fig:one}). 

\begin{figure}[h]
\tikzstyle{vertex} = [circle, draw=black, fill=black, inner sep=0pt,  minimum size=5pt]
\tikzstyle{edgelabel} = [circle, fill=white, inner sep=0pt,  minimum size=12pt]
\centering
	\pgfmathsetmacro{\d}{2}
\begin{tikzpicture}[scale=1, transform shape]
	\node[vertex, label=above:$a_i$,label=-135:] (b1) at (-2,0) {};
        \node[vertex, label=above:$b_j$, label=-45:] (x1) at (4,0) {};
        \draw [very thick, blue] (b1) arc (-120:-60:6cm)  (x1);
         \draw [very thick, orange] (b1) arc (-105:-75:11.6cm)  (x1);
        \draw [very thick, green] (x1) arc (60:120:6cm)  (b1);
         \draw [very thick, cyan] (x1) arc (75:105:11.6cm)  (b1);
        \draw [very thick, red] (b1) --  (x1);
\end{tikzpicture}
\caption{Every edge $(a_i,b_j)$ in $G'_{M'}$ gets replaced by $n_0$ parallel edges in $G^0_{M'}$ where each of these edges has a distinct color in $\{1,\ldots,n_0\}$.}
\label{fig:one}
\end{figure}
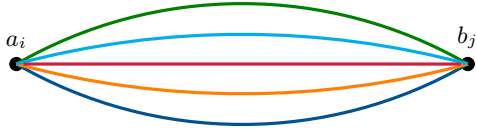

\vspace*{-0.8cm}

\begin{remark}
Recall that $|A'| = |B'|$ since $G$ admits a perfect matching, i.e., one that fully matches all vertices. Note that
$n_0 = |A'| = \sum_{a\in A}\capac(a) \le |A|\cdot|B|$ since $\capac(a) \le \delta(a) \le |B|$ for each $a \in A$ where $\delta(a)$ is $a$'s degree in $G$. 
This is because an agent has to get matched to distinct jobs in its neighborhood. Thus the size of $G^0_{M'}$ is polynomial in the size of $G$.
\end{remark}

Vertex preferences in $G^0_{M'}$ on their incident edges are as follows:
\begin{itemize}
    \item Every $a_i \in A'$ prefers any color~$c$ edge to any color~$c'$ edge where $c < c'$. Among edges of the same color, 
    $a_i$'s preference order in $G^0_{M'}$ is the same as $a_i$'s preference order in $G'_{M'}$.
    \smallskip
    \item Every $b_j \in B'$ prefers any color~$c'$ edge to any color~$c$ edge where $c < c'$. Among edges of the same color, 
    $b_j$'s preference order in $G^0_{M'}$ is the same as $b_j$'s preference order in $G'_{M'}$.
\end{itemize}

\paragraph{Stable matchings in $G^0_{M'}$.}
A matching $N$ in the multigraph $G^0_{M'}$ is a subset of $E^0_{M'}$ such that every vertex in $A' \cup B'$ has at most one edge of $N$ incident to it. An edge $e_c = (a_i,b_j)_c$ in $G^0_{M'}$ 
blocks $N$ if (i)~$a_i$ is unmatched or 
prefers $e_c$ to its $N$-edge and (ii)~$b_j$ is unmatched or prefers $e_c$ to its $N$-edge.
As before, a matching $N$ is stable if there is no edge that blocks it.

For any matching $N$ in the colorful instance $G^0_{M'}$, the matching obtained in $G'_{M'}$ by ignoring its edge colors will be called the {\em projection} of $N$ in $G'_{M'}$. 
\begin{itemize}
    \item[$\bullet$] It was shown in \cite[Theorem~2.4]{Kav24} that in the one-to-one setting, any stable matching in the colorful instance projects to a popular maximum matching in the original instance. That is,
    any stable matching in $G^0_{M'}$ projects to a popular perfect matching in~$G'_{M'}$. 
    
    \item[$\bullet$] The converse of the above was shown in \cite[Theorem~3.4]{Kav24}: for any popular maximum matching in the original instance, 
    there is a stable matching in the colorful instance whose projection is this popular maximum matching. That is, for any popular perfect matching in
    $G'_{M'}$, there is a stable matching in $G^0_{M'}$ that projects to it.
\end{itemize}

Since we are interested in properties of the matching $M'$ in $G'_{M'}$,
we state the above two results from \cite{Kav24} in the following form.

\begin{theorem}[\cite{Kav24}]
    \label{thm:old}
    $M'$ is a popular perfect matching in $G'_{M'}$ if and only if there exists a colorful version of $M'$ that 
    is stable in $G^0_{M'}$.
\end{theorem}

\subsection{A colorful many-to-many instance}
\label{sec:color-many}
Consider the many-to-many instance obtained by replacing 
every edge $e = (a,b)$ in $G$ with $n_0$ parallel copies $e_1,\ldots,e_{n_0}$.
As done earlier, we will view the $n_0$ copies of any edge $e$, i.e., $e_1,\ldots,e_{n_0}$, as $e$ colored by $n_0$ different colors: $\text{color}~1,\ldots,\text{color}~n_0$ (see Fig.~\ref{fig:one}). 
Thus we have a multigraph $G^* = (A \cup B, E^*)$ which is a {\em colorful} version of $G = (A \cup B, E)$. 
Vertex preferences in $G^*$ are analogous to vertex preferences in $G^0_{M'}$. 
\begin{enumerate}
    \item Agents prefer {\em lower} color edges to {\em higher} color edges while jobs prefer {\em higher} color edges to {\em lower} color edges. 
    \item Among edges of the same color, the preference order of any vertex on its incident edges in $G^*$ is the same as its preference order in $G$.
\end{enumerate}

\begin{definition}
A matching $M^*$ in $G^*$ is a subset of $E^*$ such that $|M^* \cap \{e_1,\ldots,e_{n_0}\}| \le 1$ for every edge $e \in E$
and $|M^*(v)| \le \capac(v)$ for every vertex $v \in A \cup B$.     
\end{definition}

\paragraph{\bf Stable matchings in $G^*$.}
Let $M^*$ be a matching in $G^*$ and let $M$ be the projection or {\em colorless} version of $M^*$.
\begin{itemize}
    \item An edge $e_c$ in $G^*$, where $e = (a,b) \notin M$,  
    is said to {\em block} $M^*$ if (i)~$a$ is not fully matched in $M^*$ or 
prefers $e_c$ to its worst $M^*$-edge and (ii)~$b$ is not fully matched in $M^*$ or 
prefers $e_c$ to its worst $M^*$-edge. 
    \item A matching $M^*$ is {\em stable} in $G^*$ if there is no edge that blocks $M^*$.
\end{itemize}

Let $M^*$ be any matching in $G^*$. Let $M^0$ be a one-to-one realization of $M^*$, i.e., for every edge $(a,b)_c \in M^*$: indices 
$i \in \{1,\ldots,\capac(a)\}$ and $j \in \{1,\ldots,\capac(b)\}$ are chosen such that neither $i$ nor $j$ has been chosen so far 
and $(a_i,b_j)_c$ is included in $M^0$. Thus the endpoints of edges in $M^0$ are in $A' \cup B'$.
Note that the same color~$c$, that was used for $(a,b)$ in $M^*$, is used for $(a_i,b_j)$ in $M^0$.

Observe that $M^0$ is a matching in $G^0_{M'}$ where $M'$ is the projection or colorless version of~$M^0$.
So $M'$ is a one-to-one realization of $M$
and $G^0_{M'}$ is the colorful one-to-one instance (from Section~\ref{sec:color-one})
corresponding to the graph $G'_{M'}$ defined in Section~\ref{sec:one-one}. 
The following lemma will be useful.

\begin{lemma}
\label{stable:many-to-one}
If $M^*$ is a stable matching in $G^*$ then $M^0$ is a stable matching in $G^0_{M'}$ 
where $M^0$ is any realization of $M^*$. 
\end{lemma}
\begin{proof}
    Suppose $M^0$ is not stable in $G^0_{M'}$. Then there is an edge $e_c$ that blocks $M^0$ where 
    $e_c = (x_i,y_j)_c$ for some $(x,y) \in E$ and $c \in \{1,\ldots,n_0\}$. 
    Note that it has to be the case that $(x,y) \notin M$ where $M$ is the projection of $M^*$. 
    Suppose $(x,y) \in M$. Then there is exactly {\em one} edge corresponding to $(x,y)$ in the graph $G'_{M'}$ and it has to be $e = (x_i,y_j)$. So there are $n_0$ edges $e_1,\ldots,e_{n_0}$ in $G^0_{M'}$. None of them can block $M^0$ because if $e_{c'} \in M^0$ 
    then agent $x_i$ prefers $e_{c'}$ to $e_{c'+1},\ldots,e_{n_0}$ (agents prefer lower color edges) while job $y_j$ prefers $e_{c'}$ to $e_1,\ldots,e_{c'-1}$ (jobs prefer higher color edges).

    So there is some edge $(x,y) \notin M$ such that for some $i \in \{1,\ldots,\capac(x)\}, 
    j \in \{1,\ldots,\capac(y)\}$, and color $c \in \{1,\ldots,n_0\}$, the edge $e_c$ in $G^0_{M'}$ blocks $M^0$ where $e = (x_i,y_j)$.
    So either $x_i$ is unmatched or prefers $e_c$ to its $M^0$-edge and similarly, either $y_j$ is unmatched or prefers $e_c$ to its $M^0$-edge. 
    This means in the instance $G^*$, there is an edge $e' = (x,y) \notin M$ and color $c \in \{1,\ldots,n_0\}$
    such that 
    either $x$ is not fully matched in $M^*$ or prefers $e'_c$ to its worst $M^*$-edge and similarly,
    either $y$ is not fully matched in $M^*$ or prefers $e'_c$ to its worst $M^*$-edge. Thus $e'_c$ blocks $M^*$, a contradiction to its stability in $G^*$. \qed
\end{proof}

As we show next, the converse of the above lemma also holds.

\begin{lemma}
\label{stable:one-to-many}
For any matching $M^*$ in $G^*$, if there is a realization  $M^0$ of $M^*$ such that $M^0$ is stable in the one-to-one instance $G^0_{M'}$
(where $M'$ is the projection of $M^0$) then $M^*$ is stable in $G^*$.
\end{lemma}
\begin{proof}
    Suppose $M^*$ is not stable in $G^*$. Then there is an edge $e \notin M$ 
    and color $c \in \{1,\ldots,n_0\}$ such that $e_c$ blocks $M^*$. Let $e = (x,y)$. 
    So either $x$ is not fully matched in $M^*$ or 
    prefers $e_c$ to its worst $M^*$-edge and similarly, either $y$ is not fully matched in $M^*$ or prefers $e_c$ to its worst $M^*$-edge. 
    If $x$ (resp., $y$) is fully matched in $M^*$ then let $e'_{c'}$ (resp., $e''_{c''}$) be the worst edge in $M^*$ that is incident to $x$ (resp., $y$); 
    let $e' = (x,y')$ and $e'' = (x',y)$. 
    
    For some $i \in \{1,\ldots,\capac(x)\}$, the one-to-one matching $M^0$ either leaves $x_i$ unmatched or contains the edge $(x_i,y'_{j'})_{c'}$ 
    where $j' \in  \{1,\ldots,\capac(y')\}$. Similarly, for some $j \in \{1,\ldots,\capac(y)\}$, the matching $M^0$ either leaves $y_j$ unmatched or
    contains the edge $(x'_{i'},y_{j})_{c''}$ where $i' \in \{1,\ldots,\capac(a')\}$. 
    Since $(x,y) \notin M$, all the $n_0$ edges $(x_i,y_j)_1,\ldots,
    (x_i,y_j)_{n_0}$ are in $G^0_{M'}$. Thus the edge $(x_i,y_j)_c$ belongs to $E^0_{M'}$ and it blocks $M^0$,
    a contradiction to its stability in $G^0_{M'}$. Hence it has to be the case that $M^*$ is stable in $G^*$. \qed
\end{proof}

Thus we can conclude the following.
\begin{itemize}
    \item  For any popular perfect matching $M$ in $G$, there exists a colorful version of $M$ that is stable in~$G^*$
          (by Theorem~\ref{prop2}, Theorem~\ref{thm:old}, and Lemma~\ref{stable:one-to-many}).
    \item Conversely, for any stable matching $M^*$ in $G^*$, the matching obtained by ignoring its edge colors is a popular perfect matching in $G$
          (by Lemma~\ref{stable:many-to-one}, Theorem~\ref{thm:old}, and Theorem~\ref{prop2}).
\end{itemize}

Summarizing, a perfect matching in $G$ is a popular perfect matching if and only if there is a {\em colorful} version of it that is stable in the colorful multigraph $G^*$. Thus solving the min-cost stable matching problem in the many-to-many instance $G^*$ solves the min-cost popular perfect matching problem in~$G$.
As mentioned earlier,
there are polynomial-time algorithms to solve the min-cost stable matching problem in a many-to-many instance~\cite{Fle03,huang2010}. 

Recall that our many-to-many instance is a {\em multigraph}. However that does not pose any additional difficulty because Fleiner's algorithm~\cite{Fle03} solves the min-cost matroid kernel problem in the intersection of two strictly ordered matroids. The matroid kernel problem 
generalizes the stable matching problem in a many-to-many instance where the underlying graph is a multigraph.\footnote{Many-to-many matchings in a simple bipartite graph are the intersection of two partition matroids while many-to-many matchings in a bipartite graph with parallel edges are the intersection of two laminar matroids.}
Thus the min-cost stable matching problem in a many-to-many instance, where the underlying graph is a multigraph, can be solved in polynomial time by Fleiner's algorithm~\cite{Fle03}. 
Hence we can conclude Theorem~\ref{thm:first} stated in Section~\ref{sec:intro}. 
We restate this result below for convenience.

\repeattheorem{restate_thm:many-many}

\section{Conclusions and open problems}
\label{sec:open}
We characterized popular perfect matchings in a many-to-many instance as projections of stable matchings in a colorful auxiliary instance. This yielded
a polynomial-time algorithm for the min-cost popular perfect matching problem in a many-to-many instance.

An interesting open problem is to show a polynomial-time algorithm for the min-cost popular {\em maximum} matching
problem in a hospitals/residents instance. Recall that a polynomial-time algorithm is known for this problem in the one-to-one 
setting. This is by a reduction to a min-cost stable matching problem in a colorful instance~\cite{Kav24}.
However this reduction does not extend to the hospitals/residents setting 
because there may exist popular maximum matchings in a hospitals/residents instance
that do not occur as stable matchings in the corresponding colorful  instance; such an example was given in \cite{KM24}. 
A more general open problem is to solve the min-cost popular maximum matching problem in a {\em many-to-many} instance.

\paragraph{\bf Acknowledgements.} Thanks to Naoyuki Kamiyama for asking us about the computational complexity of the min-cost popular maximum matching problem 
in the many-to-many setting. Telikepalli Kavitha acknowledges support by the Department of Atomic Energy, Government  of India, under project no. RTI4001. 
Kazuhisa Makino acknowledges support by JSPS KAKENHI Grant Numbers JP20H05967, JP19K22841,  JP20H00609, and the joint project of Kyoto University and Toyota Motor Corporation, titled ``Advanced Mathematical Science for Mobility Society''.

\bibliographystyle{abbrv}

\end{document}